\documentclass[aps,pra,twocolumn,10pt,groupedaddress,nofootinbib,superscriptaddress,notitlepage,showpacs,floatfix]{revtex4-1}
\usepackage{graphicx,graphics}
\usepackage{subfigure}
\usepackage{dcolumn}   % needed for some tables
\usepackage{array}
\usepackage{amsthm}
\usepackage{amsmath}
\usepackage{amssymb}
\usepackage{amsfonts}
\usepackage{mathrsfs} %curled math symbols
\usepackage{bm}% bold math
\usepackage{braket} %Dirac notation
\usepackage{enumitem}
\usepackage{url}
\usepackage[pdfstartview=FitH,pdfencoding=auto]{hyperref}
\usepackage{bookmark}
\usepackage{pifont}
\hypersetup{
    colorlinks=true,        		% false: boxed links; true: colored links
    linkcolor=blue,              	% color of internal links
    citecolor=blue,        			% color of links to bibliography
    filecolor=blue,          		% color of file links
    urlcolor=blue,              	% color of external links
    runcolor=cyan
			}
\usepackage[pdftex]{color}
\usepackage{multirow}
\usepackage{CJK}
\newcommand{\tr}{\mathrm{tr}}
\newcommand{\defeq}{:=}           			% be defined as
\newtheorem{theorem}{Theorem}
\newtheorem{lemma}{Lemma}

\begin{document}
%===================================================================================================================================================================================
% Title, authors, affiliation, dates, PACS.
%===================================================================================================================================================================================
\begin{CJK*}{GB}{}
\title{Maximally coherent states and coherence-preserving operations}
%------------------------------------------------------------------------------------------
\author{Yi Peng}%\CJKfamily{gbsn}(ÅíÒæ)}
	\thanks{These two authors contributed equally to this work.}
\author{Yong Jiang}%\CJKfamily{gbsn}(½¯ÓÂ)}
	\thanks{These two authors contributed equally to this work.}
	\affiliation{Institute of Physics, Chinese Academy of Sciences, Beijing 100190, China}
%------------------------------------------------------------------------------------------
\author{Heng Fan}%\CJKfamily{gbsn}(·¶èì)}
	\email{hfan@iphy.ac.cn}
	\affiliation{Institute of Physics, Chinese Academy of Sciences, Beijing 100190, China}
	\affiliation{Collaborative Innovation Center of Quantum Matter, Beijing 100190, China}
%------------------------------------------------------------------------------------------
\date{\today}
\eid{identifier}
%===================================================================================================================================================================================
% Abstract
%===================================================================================================================================================================================
\begin{abstract}
	We investigate the maximally coherent states to provide a refinement in quantifying coherence and give a measure-independent definition of the coherence-preserving operations.
	A maximally coherent state can be considered as the resource to create arbitrary quantum states of the same dimension by merely incoherent operations. We propose that only the
	maximally coherent states should achieve the maximal value for a coherence measure and use this condition as an additional criterion for coherence measures to obtain a
	refinement in quantifying coherence which excludes the invalid and inefficient coherence measures. Under this new criterion, we then give a measure-independent definition of the
	coherence-preserving operations, which play a similar role in quantifying coherence as that played by the local unitary operations in the scenario of studying entanglement.
\end{abstract}
\maketitle
\end{CJK*}
%===================================================================================================================================================================================
% Main contents
%===================================================================================================================================================================================
%====================================================================================================================================================================================
% Introduction
%====================================================================================================================================================================================
\section{Introduction.}
    Coherence can be considered as one of the most distinctive features of quantum mechanics. Along with quantum entanglement, quantum discord, and etc., coherence is viewed as
	a valuable resource for quantum information processing tasks~\cite{{bennett1993teleporting},{shor1997polynomial},{bennett1984quantum},{nielsen2010quantum}}, which otherwise
	can be never achieved efficiently or impossible by classical methods. Great progress has been made for quantifying entanglement and other quantum correlations from different
	viewpoints~\cite{{vedral1997quantifying},{vedral1998entanglement},{vidal2002computable},{osterloh2002scaling},{brandao2008entanglement},
	                           {amico2008entanglement},{ollivier2001quantum},{modi2012classical}}.
	However, a rigorous framework for quantifying coherence have been proposed only recently in Ref.~\cite{baumgratz2014quantifying}. Following this seminal work, fruitful
	researches have been done, some of which are mainly devoted to study the  properties of specific coherence
	measures~\cite{{shao2015fidelity},{bai2015maximally},{cheng2015complementarity},{singh2015maximally},{du2015wigner},{mani2015cohering}} or explore new possible coherence
	measures~\cite{{du2015coherence},{yuan2015intrinsic},{streltsov2015measuring}}. There are also many considerations about the manipulation of
	coherence~\cite{{yuan2015intrinsic},{streltsov2015measuring},{yao2015quantum},{bromley2015frozen},{mani2015cohering}}, and the connections of coherence with quantum
	entanglement, quantum discord, quantum deficit, and many-body systems critical phenomena~\cite{{streltsov2015measuring},{xi2015quantum},{yao2015quantum},{chen2015coherence}}.

	In this work, we would present a thorough study of the maximally coherent states (MCSs), give a refinement of quantifying coherence by adding a new criterion for valid coherence
	measures, and define the coherence-preserving operations (CPOs). It should be  notified that these three main results are closely related. The MCSs are defined
	as the states which can be used as resources to produce any other states of the same dimension by merely the incoherent (free) operations~\cite{baumgratz2014quantifying}. A
	valid coherence measure $C$ fulfilling the four criteria in	Ref.~\cite{baumgratz2014quantifying} would assign maximal value to a set of states which we may call the maximal
	coherence-value states (MCVSs) with respect to $C$. These four criteria ensure a MCS to be a MCVS for any valid coherence measures. However, for an arbitrary valid coherence
	measure, a MCVS is not necessarily a MCS. While, one may expect CPOs in quantifying coherence to play a role analogous to that	of the local unitary operations in studying
	entanglement. There was however no measure-independent definition for CPOs like that of MCSs. Instead, for a specific coherence measure $C$ we can find a set of incoherent
	operations under which the value of the coherence measure of an arbitrary state would be conserved. We may call these operations the coherence-value preserving operations
	(CVPOs) with respect to $C$. Unfortunately, the different sets of CVPOs under different valid coherence measures are not always the same. We find out that the mismatch beween
	MCSs and MCVSs happens to many inefficient coherence measures and therefore proposed a new criterion that the MCVSs should be MCSs to exclude these inefficient coherence
	measures thus give a refinement of quantifying coherence. This new criterion also makes the different groups of CVPOs of different coherence measures converge
	to the unitary incoherent operations and makes it reasonable to define the unitary incoherent operations as the CPOs. One operational meaning of this result is that coherence of
	an arbitrary states is impossible to protect in a task without the knowledge of the state to be protected and the quantum channel it would endure.

%====================================================================================================================================================================================
%	A simple review of quantifying process of coherence
%====================================================================================================================================================================================
\section{Review of quantifying coherence.}
	In quantifying coherence~\cite{baumgratz2014quantifying}, a base $\mathcal{B}{\defeq}\left\lbrace{\ket{i}}\right\rbrace$ has been chosen and fixed firstly, which would usually
	be composed of eigenstates of some conserved quantity such as the Hamiltonian of the system of interest. The quantitative theory of coherence mainly consists of three basic
	definitions and four criteria. The three definitions are:
	\begin{enumerate}[label=(D\arabic*)]
		%##########################################################################################
		\item \label{SimpleReviewOfQuantifyingCoherence_Def-IncoherentStates}
			{\it Incoherent states.} The diagonalized states in $\mathcal{B}$ are incoherent to $\mathcal{B}$. We denote the set of the incoherent states with $\mathcal{I}$.
		%##########################################################################################
		\item \label{SimpleReviewOfQuantifyingCoherence_Def-IncoherentOperations}
			{\it Incoherent operations.} Operations mapping incoherent states into incoherent states either with or without subselections are incoherent. An incoherent operation
			$\Phi_\textrm{ICPTP}$ can be specified by a set of Kraus operators $\left\lbrace{\pmb{K}_n}\right\rbrace$ with
			$\sum_{n}\pmb{K}_n^\dag\pmb{K}_n=\pmb{I}_d$, and $\pmb{\rho}_n{\in}\mathcal{I}$. It is defined that $\pmb{\rho}_n{\defeq}\pmb{K}_n\pmb{\rho}\pmb{K}_n^\dag/p_n$ and
			$p_n{\defeq}\tr\left(\pmb{K}_n\pmb{\rho}\pmb{K}_n^\dag\right)$ for all $n$. (In this work,  we would continue to use this set of notations for the incoherent operation
			$\Phi_\textrm{ICPTP}$.)
		%##########################################################################################
		\item \label{SimpleReviewOfQuantifyingCoherence_Def-MCS}
			{\it Maximally coherent states.} A MCS is one that can be used as a resource for deterministic construction of any other states of the same dimension by incoherent
			operations only. It has proven in Ref.~\cite{baumgratz2014quantifying} that $\ket{\Psi_d}{\defeq}\frac{1}{\sqrt{d}}\sum_{i=0}^{d-1}\ket{i}$ is a MCS.
		%##########################################################################################
	\end{enumerate}
	The MCSs and the incoherent states set the upper and lower bounds for coherence measures. While, the incoherent operations puts gradient in  between. To
	obtain reasonable coherence measures, four criteria have been proposed in Ref.~\cite{baumgratz2014quantifying}:
	\begin{enumerate}[label=(C\arabic*)]
		%##########################################################################################
		\item \label{SimpleReviewOfQuantifyingCoherence_Ceriterion-IncoherentVanishing}
			$C\left(\pmb{\rho}\right)=0$ if (only if) $\pmb{\rho}$ is incoherent.
		%##########################################################################################
		\item \label{SimpleReviewOfQuantifyingCoherence_Ceriterion-Monotonicity-noSubselection}
			$C\left(\Phi_\textrm{ICPTP}\left(\pmb{\rho}\right)\right){\le}C\left(\pmb{\rho}\right)$ for arbitrary $\Phi_\textrm{ICPTP}$ and $\pmb{\rho}$.
		%##########################################################################################
		\item \label{SimpleReviewOfQuantifyingCoherence_Ceriterion-Monotonicity-subselection}
			$\sum_np_nC\left(\pmb{\rho}_n\right){\le}C\left(\pmb{\rho}\right)$  for all $\Phi_\textrm{ICPTP}$ and $\pmb{\rho}$.
		%##########################################################################################
		\item \label{SimpleReviewOfQuantifyingCoherence_Ceriterion-convexity}
			The coherence measure should not increase under the mixing processes of the states.
		%##########################################################################################
	\end{enumerate}
	Any coherence measure satisfying the four criteria is considered as valid. This gives some good coherence measures such as the relative coherence measure of coherence
	$C_\textrm{rel.ent.}$, the $\ell_1$-norm coherence measure $C_{\ell_1}$, and so on.

%====================================================================================================================================================================================
%	Coherence-preserving operations and Maximally Coherent states
%====================================================================================================================================================================================
\section{Unitary incoherent operations.}
	By definition~\ref{SimpleReviewOfQuantifyingCoherence_Def-IncoherentOperations} of the incoherent operations, a CVPO of an arbitrary coherence measure is incoherent. Of all the
	incoherent operations, the unitary incoherent operations are the simplest. It would be useful and easier to examine them first.
	\begin{lemma}\label{IcoherentOperationAndMCS_lemma_IncoherentUnitary}
		All the unitary incoherent operations would take the form of
		\begin{equation}
			\pmb{U}_\textrm{I}{\defeq}\sum_{j=0}^{d-1}e^{i\theta_j}\ket{\alpha_j}\bra{j},
			 \label{IncoherentOperationAndMCS_IncoherentUnitary_def}
		\end{equation}
		where $\left\lbrace{\alpha_j}\right\rbrace$  is a relabeling of $\left\lbrace{j}\right\rbrace$. And they are CVPOs admitted by all the valid coherence measures.
	\end{lemma}
	\begin{proof}
	   We firstly prove the explicit expression of the unitary incoherent operations. Since the unitary operations would transform pure states into pure states, it is obvious that
	   the output state should be one of the base vector states, given the input is from $\mathcal{B}$. That means $\pmb{U}_\textrm{I}$ should only be a relabeling of the base
	   vectors up to some phases, namely be of the form presented in (\ref{IncoherentOperationAndMCS_IncoherentUnitary_def}). Here, we complete the proof of the first portion of
	   Lemma~\ref{IcoherentOperationAndMCS_lemma_IncoherentUnitary} and start to prove the rest by utilizing the newly proved part. One may soon realize that the inverse
	   $\pmb{U}_\textrm{I}^\dag$ is also unitary and incoherent.  Therefore, for any valid coherence measure $C$ and state $\pmb{\rho}$, we can obtain
	   $C\left(\pmb{\rho}\right){\ge}C\left(\pmb{U}_\textrm{I}\pmb{\rho}\pmb{U}_\textrm{I}^\dag\right)$ and
	   reversely  $C\left(\pmb{U}_\textrm{I}\pmb{\rho}\pmb{U}_\textrm{I}^\dag\right)
	               {\ge}C\left(\pmb{U}_\textrm{I}^\dag\left(\pmb{U}_\textrm{I}\pmb{\rho}\pmb{U}_\textrm{I}^\dag\right)\pmb{U}_\textrm{I}\right)
	               =C\left(\pmb{\rho}\right)$,
	   namely $C\left(\pmb{\rho}\right)$ and
	   $C\left(\pmb{U}_\textrm{I}\pmb{\rho}\pmb{U}_\textrm{I}^\dag\right)$ are of the same value. Thus $\pmb{U}_\textrm{I}$ is a CVPO to every valid coherence measure.
	   \qedhere
	\end{proof}

\section{Maximal coherence-value states.}
	Using Lemma~\ref{IcoherentOperationAndMCS_lemma_IncoherentUnitary}, we can obtain a set of MCVSs for every valid coherence measure by applying the unitary incoherent operations
	on $\ket{\Psi_d}$:
	\begin{equation}
		S_\textrm{MCS}{\defeq}
		 \left\lbrace{\frac{1}{\sqrt{d}}\sum_{j=0}^{d-1}e^{i\theta_j}\ket{j}{\,\Bigg|\,}\theta_1,\ldots,\theta_{d-1}{\,\in\,}\lbrack{0,2\pi})}\right\rbrace.
		\label{IncoherentOperationsAndMCS_S_MCS_def}
	\end{equation}
	Notice that we have used the ``MCS" as the subscript here because we will prove in Theorem~\ref{IncoherentOperationAndMCS_theorem_D3-MCS} that $S_\textrm{MCS}$ is the set of
	MCSs, too.

	It is very interesting but not surprising to find out that this set $S_\textrm{MCS}$ of states has its special position in the quantitative theory of coherence as a resource.
	\begin{theorem}
		$S_\textrm{MCS}$ is the complete collection of MCVSs recognised by all the valid coherence measures, as can be shown in Fig.~\ref{figure_MCS}.
		\label{IcoherentOperationAndMCS_S_MCS}
	\end{theorem}
	\begin{figure}[ht]
		\includegraphics[scale=0.60]{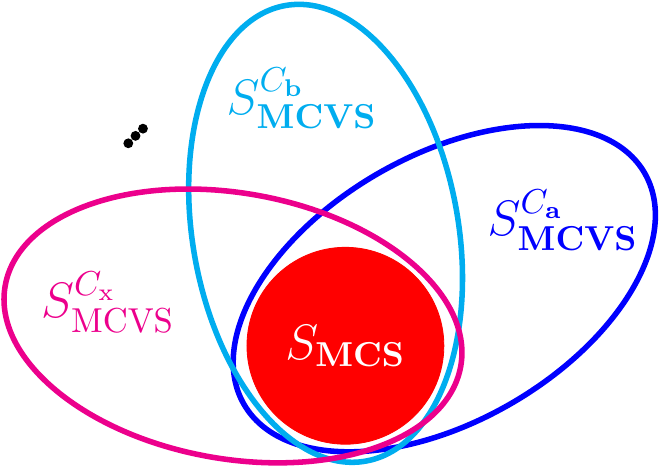}
		\caption{(Color online) Relation between the MVCSs of different valid coherence measures and $S_\textrm{MCS}$. $S_\textrm{MCVS}^C$ represent the full set of the MCVSs
			     with respect to a specific valid coherence measure $C$.}
		\label{figure_MCS}
	\end{figure}
	\begin{proof}
		We denote $S_\textrm{MCVS}$ as the complete collection of MCVSs granted by all the valid coherence measures. It is always true for any valid coherence measure $C$ that
		$S_\textrm{MCS}{\subseteq}S_\textrm{MCVS}{\subseteq}S_\textrm{MCVS}^C$. While, in Ref.~\cite{bai2015maximally} it has been shown that $S_\textrm{MCVS}^{C_\textrm{rel.ent.}}$
		coincides with $S_\textrm{MCS}$. Hence $S_\textrm{MCVS}$ should be identical to $S_\textrm{MCS}$.
		\qedhere
	\end{proof}
	However, this does not mean every valid coherence measure approves the states of $S_\textrm{MCS}$ as the sole kind of MCVSs. One can find many valid coherence measures whose
	MCVSs include more than just the states from $S_\textrm{MCS}$. A specific example is $C_\textrm{trivial}$ which is defined as a measure of the value zero iff its input state is
	incoherent otherwise always one. Another example is a continuous coherence measure $C_f$ presented in Ref.~\cite{bai2015maximally} for $d=4$. And we can follow the same way to
	construct a family of $C_f$ for arbitrary dimension $d$. Also, $S_\textrm{MCVS}^C$ could be different from one another for different $C$.

	Another fact that makes the states of $S_\textrm{MCS}$ special is that they are difficult to generate if we are constrained to using only incoherent channels. 	
	\begin{lemma}\label{IncoherentOperationAndMCS-Lemma-ICPTP-MCS}
		$\Phi_\textrm{ICPTP}\left(\pmb{\rho}\right)$ is a state of $S_\textrm{MCS}$ if and only if $\Phi_\textrm{ICPTP}$ is unitary and $\pmb{\rho}$ itself is a state in
		$S_\textrm{MCS}$.
	\end{lemma}
	\begin{proof}		
		Given $\Phi_\textrm{ICPTP}$ is unitary and $\pmb{\rho}$ belongs to $S_\textrm{MCS}$, it is apparent from Lemma~\ref{IcoherentOperationAndMCS_lemma_IncoherentUnitary} that
		$\Phi_\textrm{ICPTP}\left(\pmb{\rho}\right)$ is one of the states in $S_\textrm{MCS}$. Next, we presume that $\Phi_\textrm{ICPTP}\left(\pmb{\rho}\right)$ belongs to
		$S_\textrm{MCS}$. Then $\Phi_\textrm{ICPTP}\left(\pmb{\rho}\right)$ should be a pure state, since $S_\textrm{MCS}$ contains only pure states. This means
		$\pmb{\rho}_n=\pmb{\rho}_{n'}{\in}S_\textrm{MCS}$ for all the different $n$ and $n'$ if there are any. By clinging to this fact and using the spectral expression of
		$\pmb{\rho}$, we can finally see that $\Phi_\textrm{ICPTP}$ is unitary and $\pmb{\rho}$ is from $S_\textrm{MCS}$. To obtain this result, one may find it very helpful to
		utilize a specific property of the Kraus operator $\pmb{K}_n$ of $\Phi_\textrm{ICPTP}$ which has been stated in Ref~\cite{yao2015quantum} that there is at most one nonzero
		entry in every column of $\pmb{K}_n$. For detailed derivation, please refer to Appendix~\ref{Appendix_DetailPfLemma2}.
		\qedhere
	\end{proof}
	
\section{Maximally coherent states.}
	We will present an important result about the MCSs in the following. We have showed that the aforementioned states of $S_\textrm{MCS}$ are special as described in
	Theorem~\ref{IcoherentOperationAndMCS_S_MCS} and Lemma~\ref{IncoherentOperationAndMCS-Lemma-ICPTP-MCS}. The reason behind this is that
	\begin{theorem}
		$S_\textrm{MCS}$ is the complete set of MCSs.
		\label{IncoherentOperationAndMCS_theorem_D3-MCS}
	\end{theorem}
	\begin{proof}
		Firstly, we show that $\pmb{\rho}$ is a MCS, if $\pmb{\rho}$ belongs to $S_\textrm{MCS}$. Since the case of $\ket{\Psi_d}$ has proven explicitly in
		Ref.~\cite{baumgratz2014quantifying}, we would consider a state $\pmb{\rho}$ which is physically different from $\ket{\Psi_d}$ but still belongs to $S_\textrm{MCS}$. For
		such a state, we can transform it to $\ket{\Psi_d}$ by exploiting an unitary incoherent operation. Then use a set of incoherent operations to generate all the other states
		of the same dimension, as has been done in Ref.~\cite{baumgratz2014quantifying}. The combination of two incoherent operations can be counted still as one incoherent
		operation. Therefore, $\pmb{\rho}$ is indeed a MCS, if $\pmb{\rho}{\in}S_\textrm{MCS}$. Secondly, we would prove $\pmb{\rho}$ belongs to $S_\textrm{MCS}$, provided that
		$\pmb{\rho}$ is a MCS. If $\pmb{\rho}$ can be exploited to generate any other $d$-dimension states by incoherent processes, i.e. a MCS, we can find some
		$\Phi_\textrm{ICPTP}$ to transform it into a state of $S_\textrm{MCS}$. That means $\pmb{\rho}$ should be within $S_\textrm{MCS}$ according to
		Lemma~\ref{IncoherentOperationAndMCS-Lemma-ICPTP-MCS}. In conclusion, $\pmb{\rho}$ would fulfill the  definition~\ref{SimpleReviewOfQuantifyingCoherence_Def-MCS} of MCSs,
		iff $\pmb{\rho}{\,\in\,}S_\textrm{MCS}$.
		\qedhere
	 \end{proof}
\section{Coherence-value preserving operations.}
	We can also find out all the CVPOs admitted by every valid coherence measure.
	\begin{theorem}\label{IcoherentOperationAndMCS_S_conherence-conserve-operation-Unitary}
		The complete collection of CVPOs approved by every valid coherence measure, should be the full set of  unitary incoherent operations. This can also be expressed in
		Fig.~\ref{figure_CVPO}.
	\end{theorem}
	\begin{figure}[ht]
		\includegraphics[scale=0.60]{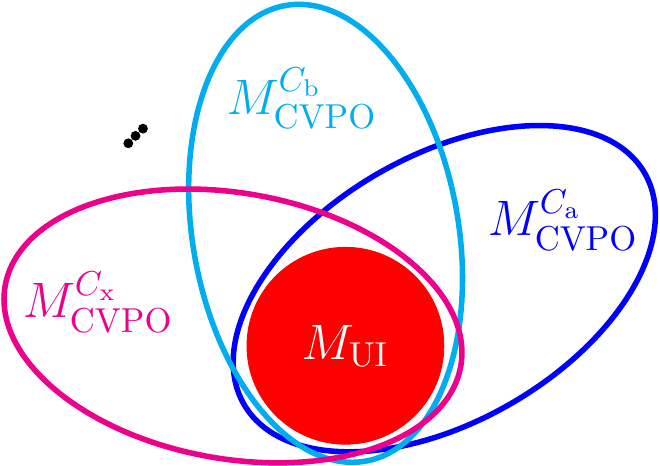}
		\caption{(Color online) Relation between the CVPOs of differen valid coherence measures and the unitary incoherent operations. Here we denote the complete collection of
				  unitary incoherent operations with $M_\textrm{UI}$, while the CVPOs for coherence measure $C$ with $M_\textrm{CVPO}^C$.}
		\label{figure_CVPO}
	\end{figure}
	\begin{proof}
		Firstly, all the unitary incoherent operations are CVPOs of an arbitrary valid coherence measure according to Lemma~{\ref{IcoherentOperationAndMCS_lemma_IncoherentUnitary}}.
		Secondly, if $\Phi_\textrm{ICPTP}$ is a CVPO admitted by every valid coherence measure and $\ket{\Psi}$ is a state of $S_\textrm{MCS}$,
		$\Phi_\textrm{ICPTP}\left(\ket{\Psi}\right)$ would be a MCVS under any measure and therefore belongs to $S_\textrm{MCS}$. By utilizing
		Lemma~\ref{IncoherentOperationAndMCS-Lemma-ICPTP-MCS}, it is clear that $\Phi_\textrm{ICPTP}$ is unitary.
		\qedhere
	\end{proof}

%====================================================================================================================================================================================
%	Cause and filtration of the ill-behaved coherence measures
%====================================================================================================================================================================================
\section{Refinement of quantifying coherence.}
	From Theorem~\ref{IcoherentOperationAndMCS_S_MCS} and Theorem~\ref{IncoherentOperationAndMCS_theorem_D3-MCS}, we can see that though
	the coherence measure satisfying the original four criteria of Ref~\cite{baumgratz2014quantifying} would count any MCS as a MCVS, many of them also give other states the maximal
	coherence values. A typical example of such inefficient but valid coherence measure is $C_\textrm{trivial}$ which as mentioned gives zero to all incoherent states but one to all
	coherent states, similar situation happens to the valid measures $C_f$~\cite{bai2015maximally} which are continuous but still inefficient. Additionally,
	Theorem~\ref{IcoherentOperationAndMCS_S_MCS} indicates that there could be differences among the sets of MCVSs of different measures. Similar disagreements exist among the sets
	of CVPOs of different coherence measures according to Theorem~\ref{IcoherentOperationAndMCS_S_conherence-conserve-operation-Unitary}. The latter would further make it difficult
	to obtain a coherence-independent definition of the CPOs. As we shall see all these problems happen to the inefficient measures such as $C_\textrm{trivial}$ and $C_f$, we
	therefore propose a new criterion for valid coherence measures to give quantifying coherence a refinement:
	\begin{enumerate}[label=(C$5$)]
		%##########################################################################################
		\item \label{CauseAndFiltrationOfTheIll-behavedCoherenceMeasure_C5}
			A valid coherence measure should only assign maximal value to the MCSs.	
		%##########################################################################################
	\end{enumerate}
	This ensures that all MCVSs are MCSs and it is the same for every coherence measure. More importantly, the inefficient coherence measures such as $C_\textrm{trivial}$ and $C_f$
	would be excluded by this newly added criterion. Some well-defined coherence measures such as the relative entropy measure, $\ell_1$-norm measure, intrinsic randomness
	measure~\cite{yuan2015intrinsic}, fulfill not only the original four criteria but also the newly added criterion. The explicit proof of
	criterion~\ref{CauseAndFiltrationOfTheIll-behavedCoherenceMeasure_C5} for these three coherence measures are provided in Appendix~\ref{AnalysMeasures}.

	Given that $C$ fulfill all the five criteria, we can use the same argument for Theorem~\ref{IcoherentOperationAndMCS_S_conherence-conserve-operation-Unitary} to show that:
	\begin{enumerate}[label=(C$5'$)]
		%##########################################################################################
		\item \label{CauseAndFiltrationOfTheIll-behavedCoherenceMeasure_C5'}
			The complete collection of CVPOs with respect to $C$ is the full set of the unitary incoherent operations.
		%##########################################################################################
	\end{enumerate}
	Therefore, the disagreements between the CVPOs of different measures would vanish, too. Moreover, \ref{CauseAndFiltrationOfTheIll-behavedCoherenceMeasure_C5'} is a necessary
	condition for all the five criteria to be fulfilled. It can be used to test if a measure can satisfy the five criteria simultaneously. A typical example is the skew information
	measure of coherence studied in Ref.~\cite{girolami2014observable}. The skew information measure would actually violate not only
	criterion~\ref{CauseAndFiltrationOfTheIll-behavedCoherenceMeasure_C5'} but also Theorem~\ref{IcoherentOperationAndMCS_S_conherence-conserve-operation-Unitary}. And this
	indicates that both \ref{SimpleReviewOfQuantifyingCoherence_Ceriterion-Monotonicity-noSubselection} one of the original four criteria and
	\ref{CauseAndFiltrationOfTheIll-behavedCoherenceMeasure_C5} the new criterion are
	violated. Our result agrees with that presented in Ref.~\cite{du2015wigner}. See detailed analysis in Appendix~\ref{AnalysMeasures}.

\section{Coherence-preserving operations.}
	An additional benefit that \ref{CauseAndFiltrationOfTheIll-behavedCoherenceMeasure_C5} provides is a natural way to define the CPOs.
	\ref{CauseAndFiltrationOfTheIll-behavedCoherenceMeasure_C5'} which is a consequence of the new criterion tells us that the set of CVPOs of any valid coherence measure $C$
	satisfying the five criteria, is independent of $C$. Furthermore, Theorem~\ref{IcoherentOperationAndMCS_S_MCS} and
	Theorem~\ref{IcoherentOperationAndMCS_S_conherence-conserve-operation-Unitary} indicate that, to all the coherence measures satisfing the original four criteria, the relation
	between the set of unitary incoherence operations and the sets of CVPOs, is structurally similar to that between $S_\textrm{MCS}$ and the different $S_\textrm{MCVS}^C$. One can
	get a clear view of this by comparing Fig.~\ref{figure_MCS} and Fig.~\ref{figure_CVPO}. For these reason we propose a definition of the CPOs.
	\begin{enumerate}[label=(D$4$)]
		%##########################################################################################
		\item \label{CPO-definition}
			An operation is coherence-preserving iff it is unitary and incoherent.
		%##########################################################################################
	\end{enumerate}
	We could see that this definition of the CPOs is measure-independent. The CPOs defined in this way are CVPOs to every coherence measure satisfying the original four criteria and
	would make a full collection of the CVPOs if the coherence measure additionally satisfies the newly added criterion.

	This result about the CPOs has one important physical implication in the general coherence preserving tasks. For an arbitrary coherence measure $C$ satisfying the five criteria,
	one may notice that the physical processes conserving the coherence values of all the $d$-dimension states could only be the processes of relabeling of the base $\mathcal{B}$.
	In other words, there is no physically non-trivial process under which the coherence value of an arbitrary $d$-dimensional state with respect to the measure $C$ can be
	conserved. However, as it is shown in Ref.~\cite{yao2015quantum}, we may find that the coherence value of some states with respect to $C$ could be frozen (conserved) under
	specific physically non-trivial processes while that of the other states could not. That means if we want the coherence value of some state to be protected, some information
	about this state and the quantum channel should be provided. Complete ignorance of the state to be protected (frozen) or the quantum channel lying ahead would make the
	protecting task impossible to achieve in principle. Moreover, by reexamining Lemma~\ref{IncoherentOperationAndMCS-Lemma-ICPTP-MCS}, we may say that MCSs are actually the most
	fragile. By that we mean the maximal coherence is the most difficult to preserve, since the only type of incoherent process preserving MCSs is relabeling.

\section{Conclusion.}
	In this work, we have provided a full collection of MCSs in (\ref{IncoherentOperationsAndMCS_S_MCS_def}), a reasonable newly added
	criterion~\ref{CauseAndFiltrationOfTheIll-behavedCoherenceMeasure_C5} for valid coherence measures
	and a measure-independent definition \ref{CPO-definition} of the CPOs. It is understandable that the states presented in (\ref{IncoherentOperationsAndMCS_S_MCS_def}) are MCSs.
	However, a valid coherence measure satisfying the original four criteria could assign maximal value to other states which are not MCSs. We therefore propose a new criterion to
	make a valid coherence measure assigns only the MCSs the maximal value and therefore exclude some inefficient coherence measures. In addition, it is apprehensible that the
	unitary incoherent operations defined in (\ref{IncoherentOperationAndMCS_IncoherentUnitary_def}) are CPOs since they are CVPOs to any coherence measure fulfilling the original
	four criteria. Similarly, other incoherent operation(s) could be CVPO(s) for some measures satisfying the original four criteria especially those with larger sets of MCVSs. With
	our newly added criterion for coherence measures, we find that only the unitary incoherent operations are the CVPOs with respect to any valid measure. We then in
	\ref{CPO-definition} identify the unitary incoherent operations as the only CPOs. Our study of the CPOs has a very significant implication that the coherence of a state is
	intrinsically hard to preserve in the case lack of information of the state and the form of the quantum channel it would undergo.

\appendix
%====================================================================================================================================================================================
% Detailed proof of Lemma~2.
%====================================================================================================================================================================================
\section{Detailed proof of Lemma~\ref{IncoherentOperationAndMCS-Lemma-ICPTP-MCS}.}\label{Appendix_DetailPfLemma2}
	We will give a detailed proof of the \emph{only if} part of  this lemma in the following.
	\begin{proof}		
		As has been claimed in the main body text that $\Phi_\textrm{ICPTP}\left(\pmb{\rho}\right){\in}S_\textrm{MCS}$ means $\pmb{\rho}_n=\pmb{\rho}_{n'}{\in}S_\textrm{MCS}$ for
		all the different $n$ and $n'$ if there are any. Notice that
		\begin{equation}
			\pmb{\rho}_n=\sum_{k}\left(q_k/p_n\right)\pmb{K}_n\ket{\varphi_k}\bra{\varphi_k}\pmb{K}_n^\dag,
			\label{IncoherentOperationsAndMCS_Lemma2_Pf_rho_n_eigen_rho_expression}
		\end{equation}
		where $q_k$ are the eigenvalues of $\pmb{\rho}$ and $\ket{\varphi_k}$ the corresponding eigenstates. One would further obtain
		\begin{equation}
			\left(\pmb{K}_n/\sqrt{p_n}\right)\ket{\varphi_k}=\left(\pmb{K}_n/\sqrt{p_n}\right)\ket{\varphi_{k'}}{\,\in\,}S_\textrm{CMS}.
			\label{Lemma-ICPTP_MCS-Kn_varphi_k_S_CMS}
		\end{equation}
		Here, we have ignored the global phase difference and will do the same in the following. This relation should be true for all $k$ and $k'$ if both $q_k$ and $q_{k'}$ are
		nonvanishing. Thus
		\begin{equation}
			\left|\braket{\varphi_k|\left(\pmb{K}_n^\dag/\sqrt{p_n}\right)|i}\right|=1/\sqrt{d},
			\label{IncoherentOperationAndMCS_D3-MCS_varphi_k|K_n|i}
		\end{equation}
		where $\ket{i}$ is an arbitrary base vector of $\mathcal{B}$. It indicates that $\pmb{K}_n^\dag\ket{i}$ should not be a null vector for any $\ket{i}$. According to
		Ref.~\cite{yao2015quantum}, if $\Phi_\textrm{ICPTP}$ is incoherent we can write $\pmb{K}_n$ as
		\begin{equation}
			\pmb{K}_n=\sum_{j=0}^{d-1}\sqrt{p_n}K_{nj}e^{i\gamma_{nj}}\ket{\lambda_{nj}}\bra{j},
		\end{equation}
		where $\left\lbrace{\lambda_{ni}}\right\rbrace=\left\lbrace{i}\right\rbrace$ and $K_{nj}$ should all be nonzero to ensure $\pmb{K}_n^\dag\ket{i}{\neq}0$. This makes
		$\pmb{K}_n$ invertible. Hence $\pmb{K}_n\ket{\varphi_k}$ would be different from $\pmb{K}_n\ket{\varphi_{k'}}$ if $\ket{\varphi_k}$ differs from $\ket{\varphi_{k'}}$.
		Applying this fact to Eq.~(\ref{IncoherentOperationsAndMCS_Lemma2_Pf_rho_n_eigen_rho_expression}), we can see that $\pmb{\rho}_n$ being a pure state implies $\pmb{\rho}$
		should also be a pure state
		\begin{equation}
			\ket{\varphi}=\sum_{j=0}^{d-1}{\varphi_j}e^{i\vartheta_j}\ket{j},
		\end{equation}
		where $\varphi_j$ are all nonnegative and satisfy the normalization condition of $\ket{\varphi}$. We can then rewrite $\pmb{\rho}_n$ as
		\begin{equation}
			 \frac{1}{\sqrt{p_n}}\pmb{K}_n\ket{\varphi}=\sum_{j=0}^{d-1}K_{nj}\varphi_je^{i\left(\gamma_{nj}+\vartheta_j\right)}\ket{\lambda_{nj}}{\in}S_\textrm{MCS}.
			 \label{Lemma-ICPTP_MCS-K_n_varphi_pure_S_MCS}
		\end{equation}
		From this expression, we can know that there is no null $\varphi_j$ and $K_{nj}=1/\left(\sqrt{d}\varphi_j\right)$. $K_{nj}$ is thus independent of $n$. Also,
		$\left(\gamma_{nj}-\gamma_{nj'}\right)$ should be independent of $n$ for every $j$ and $j'$, because $\pmb{\rho}_n=\pmb{\rho}_{n'}$ namely
		$\left(\pmb{K}_n/\sqrt{p_n}\right)\ket{\varphi}=\left(\pmb{K}_{n'}/\sqrt{p_{n'}}\right)\ket{\varphi}$. Therefore, $\left(\pmb{K}_n/\sqrt{p_n}\right)$ and
		$\left(\pmb{K}_{n'}/\sqrt{p_{n'}}\right)$ are mutually equivalent up to some global phase. One may notice that the Kraus operators $\pmb{K}_n$ have been considered are those
		with nonzero $p_n$. It is enough though. Given the facts that diagonal entries of of the sum of $\pmb{K}_n^\dag\pmb{K}_n$ with nonvanishing $p_n$ should never exceed one and
		there is a normalization constraint on $\varphi_j$, we can obtain that $\pmb{\rho}$ should belong to $S_\textrm{MCS}$. And $\Phi_\textrm{ICPTP}$ would be an unitary 
		operation given further the completeness relation of the Kraus operators.
		\qedhere
	\end{proof}

%====================================================================================================================================================================================
%	Analysis Of Explicit coherence measures
%====================================================================================================================================================================================
	\section{Analysis of specific coherence measures}\label{AnalysMeasures}
	In this section, we would firstly analyze some coherence measures satisfying the original four
	criteria~\ref{SimpleReviewOfQuantifyingCoherence_Ceriterion-IncoherentVanishing}$\sim$\ref{SimpleReviewOfQuantifyingCoherence_Ceriterion-convexity} and show that they satisfy
	also the newly added criterion~\ref{CauseAndFiltrationOfTheIll-behavedCoherenceMeasure_C5}. Among them, the relative entropy measure and $\ell_1$-norm measure have proven in
	Ref~\cite{baumgratz2014quantifying}, while the intrinsic randomness measure in Ref~\cite{yuan2015intrinsic}, to fulfill the original four criteria. And we would also discuss the
	skew information which was claimed to satisfy the original four criteria in Ref~\cite{girolami2014observable}. The skew information measure turn out to violate not only our
	newly added criterion~\ref{CauseAndFiltrationOfTheIll-behavedCoherenceMeasure_C5} but also one of the original four
	criterion~\ref{SimpleReviewOfQuantifyingCoherence_Ceriterion-Monotonicity-noSubselection} in the general case of $d{\ge}3$. 	
	%################################################################################################################################################################################
	%	Relative entropy coherence measure
	%################################################################################################################################################################################
	\subsection{Relative entropy coherence measure}
	$C_\textrm{rel.ent.}$ can certainly fulfill criterion~\ref{CauseAndFiltrationOfTheIll-behavedCoherenceMeasure_C5}, because the full set of maximal coherence-value states
	$S_\textrm{MCVS}^{C_\textrm{rel.ent.}}$ is identical to $S_\textrm{MCS}$ as has been presented in Ref.~\cite{bai2015maximally}.
	%################################################################################################################################################################################
	%	$\ell_1$-norm coherence measure
	%################################################################################################################################################################################
		\subsection{\texorpdfstring{$\ell_1$}{\texttwoinferior}-norm coherence measure}
		For the $\ell_1$ coherence measure we can show that it satisfies \ref{CauseAndFiltrationOfTheIll-behavedCoherenceMeasure_C5}, too. We can obtain the maximal value of
		$\ell_1$-norm measure of coherence $C_{\ell_1}\left(\ket{\Psi_d}\bra{\Psi_d}\right)=d-1$, given $\ket{\Psi_d}=\frac{1}{\sqrt{d}}\sum_{i=0}^{d-1}\ket{i}$ and
		\begin{equation}
			C_{\ell_1}\left(\pmb{\rho}\right)=\sum_{\stackrel{i,j=0}{i{\neq}j}}^{d-1}\left|\braket{i|\pmb{\rho}|j}\right|.
		\end{equation}
		One may consider an arbitrary state
		\begin{equation}
			\pmb{\rho}=\sum_{k}q_k\ket{\varphi_k}\bra{\varphi_k},
		\end{equation}
		where all the $q_k$ are positive and fulfill the trace normalization condition. It can be derived that
		\begin{eqnarray}
			C_{\ell_1}\left(\pmb{\rho}\right)
			&=&		\sum_{j,j'=0}^{d-1}\left|\braket{j|\pmb{\rho}|j'}\right|-1												\nonumber\\
			&=&		\sum_{j,j'=0}^{d-1}\left|\sum_{k}q_k\braket{j|\varphi_k}\braket{\varphi_k|j'}\right|-1					\nonumber\\
			&{\le}&	\sum_{j,j'=0}^{d-1}\sum_{k}q_k\left|\braket{j|\varphi_k}\right|\left|\braket{\varphi_k|j'}\right|-1	
			\label{AnalysisOfExpliciteCoherenceMeasures_ell1-inequality_phase-contraint}									\\
			&=&		d^2\sum_{k}q_k\left(\sum_{j=0}^{d-1}\frac{1}{d}\left|\braket{j|\varphi_k}\right|\right)^2-1				\nonumber\\
			&{\le}& d^2\sum_{k}q_k\sum_{j=0}^{d-1}\frac{1}{d}\left|\braket{j|\varphi_k}\right|^2-1	
			\label{AnalysisOfExpliciteCoherenceMeasures_ell1-inequality_amplitude-contraint}								\\
			&=&		d-1.
		\end{eqnarray}
		As we shall see, to make the equality in (\ref{AnalysisOfExpliciteCoherenceMeasures_ell1-inequality_amplitude-contraint}) hold true, it is required that
		$\left|\braket{j|\varphi_k}\right|$ must all be of the same value $1/\sqrt{d}$. Therefor, we can give $\ket{\varphi_k}$ the following expression
		\begin{equation}
			\ket{\varphi_k}=\frac{1}{\sqrt{d}}\sum_{j=0}^{d-1}e^{i\theta_{kj}}\ket{j}.
		\end{equation}
		To further reduce the inequality (\ref{AnalysisOfExpliciteCoherenceMeasures_ell1-inequality_phase-contraint}) into an equality, we must make sure that either there is only
		one nonzero $q_k$ or $\braket{j|\varphi_k}\braket{\varphi_k|j'}=e^{i\left(\theta_{kj}-\theta_{kj'}\right)}/d$ is independent of $k$. That means $\pmb{\rho}$ is a pure state
		and must come from $S_\textrm{MCS}$. And one may notice that the $\ell_1$-norm coherence measure of a state from $S_\textrm{MCS}$ would always be $(d-1)$. Therefore, the
		$\ell_1$-norm measure of coherence also satisfies \ref{CauseAndFiltrationOfTheIll-behavedCoherenceMeasure_C5}.
	%################################################################################################################################################################################
	%	Intrinsic randomness
	%###############################################################################################################################################################################
	\subsection{Intrinsic randomness}
		It is in Ref.~\cite{yuan2015intrinsic} that the so-called intrinsic randomness has been defined as
		\begin{equation}
			C_\textrm{int.rand.}\left(\pmb{\rho}\right)
			{\defeq}\left\lbrace
			\begin{array}{rl}
				C_\textrm{rel.ent.}\left(\pmb{\rho}\right),
				&	\textrm{if}{\,\,}\pmb{\rho}{\,\,}\textrm{is pure},		\\
				\min\limits_{q_k,\pmb{\rho}_k}\sum\limits_kq_kC_\textrm{rel.ent.}\left(\pmb{\rho}_k\right),
				&   \textrm{otherwise}.
			\end{array}
			\right.
		\end{equation}
		Now, we set to prove that it also satisfies the newly added criterion~\ref{CauseAndFiltrationOfTheIll-behavedCoherenceMeasure_C5}. Firstly when $\pmb{\rho}$ is pure, the
		intrinsic randomness measure coincides with the relative entropy measure. $C_\textrm{int.rand.}\left(\pmb{\rho}\right)$ can therefore achieve the maximal value iff
		$\pmb{\rho}$ is within $S_\textrm{MCS}$. While in the
		case that $\pmb{\rho}$ is a mixed state, $C_\textrm{int.rand.}\left(\pmb{\rho}\right)$ could be of that maximal value only if $\pmb{\rho}$ can be decomposed solely into
		statistical mixture of states from $S_\textrm{MCS}$, which is however not possible. Because a mixed state always has at least two distinct eigenvectors $\ket{\varphi_0}$ and
		$\ket{\varphi_1}$ with nonvanishing eigenvalues $q_0$ and $q_1$. For convenience of discussion we can assume $q_0{\le}q_1$ without loss of generality. One may realize that
		$\left(q_0\ket{\varphi_0}\bra{\varphi_0}+q_1\ket{\varphi_1}\bra{\varphi_1}\right)$ can be replaced by
		 $\left\lbrack{q_0\ket{\varphi_{+}}\bra{\varphi_{+}}+q_0\ket{\varphi_{-}}\bra{\varphi_{-}}+\left(q_1-q_0\right)\ket{\varphi_1}\bra{\varphi_1}}\right\rbrack$. The states
		$\ket{\varphi_\pm}$ are defined as some superpositions of $\ket{\varphi_0}$ and $\ket{\varphi_1}$. And they are designed to be mutually orthogonal. By choosing
		the superposition parameters carefully, we can keep $\ket{\varphi_{\pm}}$ out of $S_\textrm{MCS}$ even if $\ket{\varphi_0}$ and $\ket{\varphi_1}$ belong to $S_\textrm{MCS}$.
		That means a mixed state can never only have decompositions of states from $S_\textrm{MCS}$. Thus, $\pmb{\rho}$ is a MCVS with respect to the intrinsic randomness measure of
		coherence iff $\pmb{\rho}{\in}S_\textrm{MCS}$.
	%################################################################################################################################################################################
	% Skew information
	%################################################################################################################################################################################
	\subsection{Skew information}
	The skew information~\cite{{wigner1963information},{girolami2014observable}} is defined as
		\begin{equation}
			C_\textrm{skew}\left(\pmb{\rho},\pmb{K}\right)
			{\defeq}-\frac{1}{2}\tr\left(\left\lbrack{\sqrt{\pmb{\rho}},\pmb{K}}\right\rbrack^2\right)
		\end{equation}
		where $\pmb{K}{\defeq}\sum_{i=0}^{d-1}k_i\ket{i}\bra{i}$ is self-adjoint and $k_i{\neq}k_j$ for different $i$ and $j$. For a pure state $\pmb{\rho}=\ket{\psi}\bra{\psi}$,
		we can find that
		\begin{eqnarray}
			&& C_\textrm{skew}\left(\ket{\psi}\bra{\psi},\pmb{K}\right)																	\nonumber\\
			&=&\braket{\psi|\pmb{K}^2|\psi}-\left(\braket{\psi|\pmb{K}|\psi}\right)^2													\nonumber\\
			&=&\sum_{i=0}^{d-1}k_i^2\left|\braket{i|\psi}\right|^2-\left(\sum_{i=0}^{d-1}k_i\left|\braket{i|\psi}\right|^2\right)^2		\nonumber\\
			&=&\frac{1}{2}\sum_{\stackrel{i,j=0}{i{\neq}j}}^{d-1}\left|\braket{i|\psi}\right|^2\left(k_i-k_j\right)^2\left|\braket{j|\psi}\right|^2.
			\label{skewInformation_pureStates}
		\end{eqnarray}
		We can now see that $C_\textrm{skew}\left(\ket{\psi}\bra{\psi},\pmb{K}\right)$ would not be conserved under an unitary incoherent operations, i.e. a relabeling of the base
		vectors up to some phases, given $d{\ge}3$. Therefore, Theorem~\ref{IcoherentOperationAndMCS_S_conherence-conserve-operation-Unitary} and
		\ref{CauseAndFiltrationOfTheIll-behavedCoherenceMeasure_C5'} would be violated. Also, we know that
		Theorem~\ref{IcoherentOperationAndMCS_S_conherence-conserve-operation-Unitary} is a consequence of
		\ref{SimpleReviewOfQuantifyingCoherence_Ceriterion-Monotonicity-noSubselection} one of the original four criteria while
		\ref{CauseAndFiltrationOfTheIll-behavedCoherenceMeasure_C5'} is the consequence of \ref{CauseAndFiltrationOfTheIll-behavedCoherenceMeasure_C5} our newly
		added criteria. Hence, neither  \ref{SimpleReviewOfQuantifyingCoherence_Ceriterion-Monotonicity-noSubselection} nor
		\ref{CauseAndFiltrationOfTheIll-behavedCoherenceMeasure_C5} would be fulfilled.

%====================================================================================================================================================================================
%	Acknowledge
%====================================================================================================================================================================================
\begin{acknowledgments}
	This work is supported by NSFC (No.91536108), grant from Chinese Academy of Sciences (XDB01010000).	We thank useful discussions with C. P. Sun, M. Gu and V. Vedral.
\end{acknowledgments}

\bibliography{Bibliography}
\end{document}